\documentclass[11pt]{article}

\usepackage{a4}
\usepackage[top=30truemm,bottom=30truemm]{geometry}

\usepackage{authblk}
\usepackage{graphicx}
\usepackage{rotating}
\usepackage{comment}
\usepackage{multirow}

\usepackage{amsmath}
\usepackage{amssymb}
\usepackage{amsthm}
\usepackage{algorithm}
\usepackage[noend]{algpseudocode}
\usepackage{algorithmicx}
\usepackage{url}
\usepackage{enumerate}
\usepackage{bm}
\usepackage{setspace}
\usepackage{epsfig}
\usepackage{color}

\usepackage[draft,mode=multiuser]{fixme}

\fxuseenvlayout{colorsig}
\fxusetargetlayout{color}
\FXRegisterAuthor{si}{asi}{SI}
\fxsetface{env}{}

\newcommand{\Suffix}{\mathsf{Suffix}}
\newcommand{\Substr}{\mathsf{Substr}}
\newcommand{\STree}{\mathsf{STree}}

\newcommand{\suflink}{\mathsf{slink}}

\newcommand{\rev}[1]{#1^R}

\newcommand{\occ}{\mathsf{occ}}
\newcommand{\LM}{\mathsf{LM}}
\newcommand{\RM}{\mathsf{RM}}
\newcommand{\M}{\mathsf{M}}

\newcommand{\PNF}{\mathsf{NF^+}}

\def\W#1#2{{\mathsf{W\_link}}_{#1}({#2})}

\newtheorem{theorem}{Theorem}
\newtheorem{lemma}{Lemma}
\newtheorem{corollary}{Corollary}

\theoremstyle{definition}

\newtheorem{example}{Example}

\newtheorem{problem}{Problem}

\title{
  Faster and Simpler Online Computation \\ of String Net Frequency
}

\author{Shunsuke Inenaga}

\affil{\textit{Department of Informatics, Kyushu University, Japan} \\
\texttt{\small inenaga.shunsuke.380@m.kyushu-u.ac.jp}
}

\date{}

\begin{document}
\maketitle

\begin{abstract}
  An occurrence of a repeated substring $u$ in a string $S$
  is called a \emph{net occurrence} if extending
  the occurrence to the left or to the right decreases
  the number of occurrences to 1.
  The \emph{net frequency} (\emph{NF}) of a repeated substring $u$ in a string $S$
  is the number of net occurrences of $u$ in $S$.
  Very recently, Guo et al. [SPIRE 2024] proposed an online $O(n \log \sigma)$-time algorithm that maintains a data structure of $O(n)$ space which answers \textsc{Single-NF} queries in $O(m\log \sigma + \sigma^2)$ time
  and reports all answers of the \textsc{All-NF} problem in $O(n\sigma^2)$ time.
  Here, $n$ is the length of the input string $S$,
  $m$ is the query pattern length, and $\sigma$ is the alphabet size.
  The $\sigma^2$ term is a major drawback of their method since
  computing string net frequencies is originally motivated for Chinese language text processing where $\sigma$ can be thousands large.
  This paper presents an improved online $O(n \log \sigma)$-time algorithm,
  which answers \textsc{Single-NF} queries in $O(m \log \sigma)$ time and reports all answers to the \textsc{All-NF} problem in output-optimal $O(|\PNF(S)|)$ time, where $\PNF(S)$ is the set of substrings of $S$ paired with their positive NF values. We note that $|\PNF(S)| = O(n)$ always holds.
  In contrast to Guo et al.'s algorithm that is based on Ukkonen's suffix tree construction, our algorithm is based on Weiner's suffix tree construction.
\end{abstract}


\section{Introduction}

A \emph{repeat} in string $S$ is a substring that occurs at least twice in $S$.
An occurrence $[i..j]$ of a repeat $u$ in $S$ is said to be
a \emph{net occurrence} of $u$ if neither of its left-extension nor right-extension is a repeat, i.e.,
each of $au$, $ub$, and $aub$ occurs exactly once in $S$
where $a = T[i-1]$, $b = T[j+1]$, and $1 < i \leq j < n$.
The number of net occurrences of a repeat $u$ in a string $S$ is said to be the \emph{net frequency} (\emph{NF}) of $u$ in $S$.
Computing the string net frequencies has applications in
Chinese language text processing~\cite{LinY01,LinY04,OhlebuschBO24}.

Efficient computation of string net frequencies has gained recent attention.
Guo, Eades Wirth, and Zobel~\cite{GuoEWZ24}
proposed the first efficient algorithm for computing the string net frequencies
(the \emph{GEWZ} algorithm).
The GEWZ algorithm works in an offline manner,
namely it first reads the whole input string $S$,
and builds a data structure of $O(n)$ space based on
the suffix array~\cite{manber93:_suffix} and the Burrows-Wheeler transform~\cite{BurrowsWheeler94}.
The data structure of the GEWZ algorithm takes $O(n)$ total space,
and can be built in $O(n)$ time 
for integer alphabets of polynomial size in $n$,
or in $O(n \log \sigma)$ time for general ordered alphabets of size $\sigma$.
Then, the GEWZ algorithm answers \textsc{Single-NF} queries
of reporting the net frequency of a given pattern in $S$
in $O(m + \log n + \sigma)$ time each\footnote{Theorem 15 of the GEWZ algorithm~\cite{GuoEWZ24} assumes that a pair $(i,i+m-1)$ of positions for an occurrence $[i..i+m-1]$ of a pattern in text $S$ is given as a query, where $m$ is the pattern length. When a pattern is given as a string, then since their data structure is based on the suffix array, it requires $O(m + \log n)$ time for finding an occurrence.}.
They also considered the \textsc{All-NF} problem of reporting
the net frequencies of all repeats in $S$ that have positive net frequencies,
and showed how their data structure can answer
the \textsc{All-NF} problem in $O(n)$ time.
Ohlebusch, B{\"{u}}chler, and Olbrich~\cite{OhlebuschBO24} proposed
another offline algorithm (the OBO algorithm) for the \textsc{All-NF} problem that works in $O(n)$ total time for integer alphabets of size $\sigma = O(n)$.

Recently, Guo, Umboh, Wirth, and Zobel~\cite{GuoUWZ24} considered the online version of the problem of computing net frequency, where the characters of the input string arrive one by one.
The GUWZ algorithm given in~\cite{GuoUWZ24} uses
the suffix tree~\cite{Weiner} of $O(n)$ space and answers \textsc{Single-NF} queries in $O(m \log \sigma + \sigma^2)$ time\footnote{The alphabet size $\sigma$ is assumed to be a constant in the literature~\cite{GuoEWZ24,GuoUWZ24}.}.
For the online \textsc{All-NF} problem, 
the GUWZ algorithm maintains a data structure for an input online string $S$
that takes $O(n \sigma^2)$ time to report the answers to the \textsc{All-NF} problem.
The GUWZ algorithm
is based on Ukkonen's left-to-right online suffix tree construction algorithm~\cite{Ukkonen95} and uses $O(\log \sigma)$ amortized time for updates per new character.

There are two major drawbacks in the GUWZ algorithm.
The first drawback is its dependency in the alphabet size.
The $\sigma^2$ terms in the $O(m \log \sigma + \sigma^2)$ query time
for the \textsc{Single-NF} problem
and in the $O(n\sigma^2)$ reporting time for the \textsc{All-NF} problem
can be a major bottleneck in Chinese language text processing
in which the alphabet size $\sigma$ is thousands large.
The second drawback is that the GUWZ algorithm heavily relies on the
suffix-extension data structure of Breslauer and Italiano~\cite{BreslauerI12},
which further relies on a nearest marked ancestor (NMA) data structure~\cite{westbrook92:_fast_increm_planar_testin} on dynamic trees and
an order-maintenance data structure~\cite{Tsakalidis84,DietzS87,BenderCDFZ02} on dynamic lists.

In this paper, we resolve both of the aforementioned issues
by presenting an improved algorithm for the
\textsc{Single-NF} and \textsc{All-NF} problems in the online setting (see also Table~\ref{tab:comparisons}).
We present an online algorithm that uses $O(n)$ space
and maintains a list of repeats with positive NF values in $O(\log \sigma)$ amortized time per character, for general ordered alphabets of size $\sigma$.
Since the positive NF values can be stored in the corresponding nodes
in the suffix tree with $O(n)$ space,
our data structure answers \textsc{Single-NF} queries in $O(m \log \sigma)$ time each.
Also, our algorithm maintains a data structure that
can report all answers to the \textsc{All-NF} problem in output optimal time,
independently of the alphabet size $\sigma$.

Our algorithm is based on Weiner's suffix tree construction~\cite{Weiner},
that works in a right-to-left online manner.
For ease of explanations, we use the version of Weiner's algorithm
that uses \emph{soft Weiner links} on the suffix tree,
which are secondary edges of the directed acyclic word graph (DAWG)~\cite{Blumer85} of the reversed string.


Our online NF computation algorithm is built on
deep insights to how Weiner's online suffix tree construction algorithm
maintains structural changes in the suffix tree in the online setting.
Notably, no special data structures or operations are needed
for computing the string NFs online -
all the necessary information is computed by the classical Weiner algorithm
and the string NFs can be obtained almost for free during the online computation of the suffix tree.
We note that our running time $O(n \log \sigma)$ is optimal
as long as we maintain a data structure based on suffix trees
in the online string~\cite{Farach-ColtonFM00}.

Very recently,
Mieno and Inenaga~\cite{MienoI25}
proposed another online NF algorithm that is based on
Ukkonen's left-to-right suffix tree construction~\cite{Ukkonen95}.
While the MI algorithm of~\cite{MienoI25} requires some non-trivial tweaks
for maintaining the \emph{secondary active point}
on top of the usual active point over the online suffix tree,
no such additional mechanism is required in our Weiner-based method.

\begin{table}[htb]
  \centering{
  \label{tab:comparisons}
  \caption{Comparison of algorithms for computing string net frequencies for input string $S$ of length $n$ over an alphabet of size $\sigma$. $\PNF(S)$ denotes the set of substrings of $S$ paired with their positive NF values, where $|\PNF(S)| = O(n)$ always holds. The $\log \sigma$ factors in the construction times of the offline algorithms~\cite{GuoEWZ24,OhlebuschBO24} can be removed for integer alphabets of size $\sigma = n^{O(1)}$.}
\vspace*{2mm}
\begin{tabular}{c||c|c|c|c|c} \hline
  algorithm & construction time & space & \textsc{Single-NF} time & \textsc{All-NF} time & model \\ \hline \hline
  GEWZ~\cite{GuoEWZ24} & $O(n \log \sigma)$ & $O(n)$ & $O(m + \log n + \sigma)$ & $O(n)$ & offline \\ \hline
  OBO~\cite{OhlebuschBO24} & $O(n \log \sigma)$ & $O(n)$ & - & $O(n)$ & offline \\ \hline
  GUWZ~\cite{GuoUWZ24} & $O(n \log \sigma)$ & $O(n)$ & $O(m \log \sigma + \sigma^2)$ & $O(n\sigma^2)$ & online \\ \hline
  Our Algorithm & $O(n \log \sigma)$ & $O(n)$ & $O(m \log \sigma)$ & $O(|\PNF(S)|)$ & online \\ \hline
  MI~\cite{MienoI25} & $O(n \log \sigma)$ & $O(n)$ & $O(m \log \sigma)$ & $O(|\PNF(S)|)$ & online \\ \hline
\end{tabular}
}
\end{table}

\section{Preliminaries}

\subsection{Strings}

Let $\Sigma$ be a general ordered alphabet.
Any element of $\Sigma^*$ is called a \emph{string}.
For any string $S$, let $|S|$ denote its length.
Let $\varepsilon$ be the empty string, namely, $|\varepsilon| = 0$.
Let $\Sigma^+ = \Sigma \setminus \{\varepsilon\}$.
If $S = XYZ$, then $X$, $Y$, and $Z$ are called 
a \emph{prefix}, a \emph{substring}, and a \emph{suffix} of $S$, respectively.
For any $1 \leq i \leq j \leq |S|$,
let $S[i..j]$ denote the substring of $S$ that begins at position $i$
and ends at position $j$ in $S$.
For convenience, let $S[i..j] = \varepsilon$ if $i > j$.
For any $1 \leq i \leq |S|$, let $S[i]$ denote the $i$th character of $S$.
For any string $T$, let
$\Substr(S)$ and $\Suffix(S)$ denote the sets of substrings and suffixes of $S$,
respectively.
For any string $S$, let $\rev{S}$ denote the reversed string of $S$,
i.e., $\rev{S} = S[|S|] \cdots S[1]$.

\subsection{String Net Frequency and Maximal Repeats}

For a string $w$ and a string $S$,
let $\occ_S(w)$ denote the number of occurrences of $w$ in $S$,
i.e. $$\occ_S(w) = |\{i \mid S[i..i+|w|-1] = w\}|.$$

A string $w$ with $\occ_S(w) \geq 2$ is called a \emph{repeat} in $S$.
For a repeat $w \in \Substr(S)$ of a string $S$,
let
\begin{equation}
\Phi_S(w) = \{(\alpha,\beta) \in \Sigma \times \Sigma \mid \occ_S(\alpha w \beta) = \occ_S(\alpha w) = \occ_S(w \beta) = 1\}. \label{eqn:Phi}
\end{equation}
The \emph{net frequency} (\emph{NF}) of $w$ in $S$,
denoted $\phi_S(w),$ is the number of
ordered pairs $(\alpha, \beta) \in \Sigma \times \Sigma$ of characters such that
$\occ_S(\alpha w \beta) = \occ_S(\alpha w) = \occ_S(w \beta) = 1$, i.e.,
\begin{equation}
\phi_S(w) = |\Phi_S(w)|.
\label{eqn:phi}
\end{equation}

A repeat $w$ in a string $S$ is called
a \emph{left-maximal repeat} in $S$
if $w$ is a prefix of $S$,
or there are at least two distinct characters $\alpha, \alpha'$
that immediately precede $w$ in $S$
(i.e. $\alpha w, \alpha'w \in \Substr(S)$).
Similarly,
a repeat $w$ in $S$ is called a \emph{right-maximal repeat} in $S$
if $w$ is a suffix of $S$,
or there are at least two distinct characters $\beta, \beta'$
that immediately follow $w$ in $S$
(i.e. $w \beta, w \beta' \in \Substr(S)$).
Let $\LM(S)$ and $\RM(S)$ denote the sets of
left-maximal and right-maximal repeats in $S$, respectively.

For a left-maximal repeat $u$ in $S$,
a character $\alpha \in \Sigma$ such that $\occ_S(\alpha w) \geq 1$
is called a \emph{left-extension} of $w$ in $S$.
Similarly, for a right-maximal repeat $w$ in $S$,
a character $\beta \in \Sigma$ such that $\occ_S(w \beta) \geq 1$
is called a \emph{right-extension} of $w$ in $S$.

A repeat $u$ is called a \emph{maximal repeat} in $S$
if $w \in \LM(S) \cap \RM(S)$.
Let $\M(S) = \LM(S) \cap \RM(S)$ denote the set of maximal repeats in $S$.

Guo et al.~\cite{GuoEWZ24} showed that
if $\phi_S(w) \geq 1$, then $w \in \RM(S)$.
Below we present a stronger statement.

\begin{lemma} \label{lem:NF_MR}
If $\phi_S(w) \geq 1$, then $w \in \M(S)$. 
\end{lemma}
\begin{proof}
  Since $\phi_S(w) \geq 1$,
  there exists a pair $(\alpha,\beta)$ of characters
  such that $\occ_S(\alpha w) = \occ_S(w \beta) = 1$.
  Since $\occ_S(w) \geq 2$,
  there is an occurrence $w = S[i..i+|w|-1]$ of $w$
  such that $S[i-1] \neq \alpha$ or $i = 1$.
  Similarly, there is an occurrence $w = S[j-|w|+1..j]$ of $w$
  such that $S[j+1] \neq \beta$ or $j = |S|$.
  Thus $w \in \LM(S) \cap \RM(S) = \M(S)$.
\end{proof}


Let
$$\PNF(S) = \{(w, \phi_S(w)) \mid w \in \M(S), \phi_S(w) \geq 1\}$$
denote the set of 
pairs of string $w$ and its NF that is positive in $S$.
It follows from Lemma~\ref{lem:NF_MR} that if $(w, \phi_S(w)) \in \PNF(S)$,
then $w \in \M(S)$.
Since $|\M(S)| \leq n-1$ for any string of length $n$ (c.f.~\cite{Blumer87}),
it is immediate that $|\PNF(S)| \leq n-1$.

\subsection{Online Computation of String Net Frequency} \label{sec:online_NF_def}

We will design an online algorithm for computing
substrings that have positive values of net frequencies,
in which new characters are \emph{prepended} to the input string.
Namely, our algorithm works in a right-to-left online manner.
Note that $\phi_S(w)$ is closed under reversal, i.e.,
$\phi_S(w) = \phi_{\rev{S}}(\rev{w})$. 
Hence, by reversing the input string,
our algorithm can be used in the left-to-right online setting as well.

We consider the following problem of maintaining
the list $\mathcal{L}$ that stores $\PNF(S)$ for an online string $S$:


\begin{problem}[Right-to-Left Online NF Computation] \label{prob:online_NF_right_to_left}
  Let $S$ be a string of length $n$
  for which $\PNF(S)$ has been computed.
  Given a character $a \in \Sigma$,
  update $\mathcal{L}$ by computing
  $\PNF(aS) \Delta \PNF(S) = \left(\PNF(aS) \setminus \PNF(S)\right) \cup \left(\PNF(S) \setminus \PNF(aS)\right)$.
\end{problem}


Problem~\ref{prob:online_NF_right_to_left} 
requires one to consider the following cases
each time a new character $a \in \Sigma$ is prepended to the current string $S$:
\begin{description}
 \item[Case (i)] Increase the NFs of some existing maximal repeats: For each $w \in \M(aS) \cap \M(S)$
  such that $(w, \phi_{S}(w)) \in \PNF(S)$ and $\phi_{aS}(w) > \phi_{S}(w) \geq 1$,
  replace $(w, \phi_S(w))$ with $(w, \phi_{aS}(w))$ in the list $\mathcal{L}$.
  Also, for each $w \in \M(aS) \cap \M(S)$
  such that $(w, \phi_{S}(w)) \notin \PNF(S)$ and $\phi_{aS}(w) > \phi_{S}(w) = 0$,
  add $(w, \phi_{aS}(w))$ in the list $\mathcal{L}$.
 \item[Case (ii)] Add new maximal repeats whose NFs become positive: For each $w \in \M(aS) \setminus \M(S)$ such that $(w, \phi_S(w)) \notin \PNF(S)$ and $\phi_{aS}(w) > \phi_S(w) = 0$,
  add $(w, \phi_{aS}(w))$ to the list $\mathcal{L}$.
 \item[Case (iii)]
 Decrease the NFs of some existing maximal repeats: For each $w \in \M(aS) \cap \M(S)$
 such that $(w, \phi_{S}(w)) \in \PNF(S)$ and $\phi_{aS}(w) < \phi_{S}(w)$,
 replace $(w, \phi_S(w))$ with $(w, \phi_{aS}(w))$ in the list $\mathcal{L}$
 if $\phi_{aS}(w) \geq 1$,
 and remove $(w, \phi_S(w))$ from $\mathcal{L}$
 if $\phi_{aS}(w) = 0$.
\end{description}

\begin{example}
  Consider string $S = \mathtt{ababbababcababbb\$}$
  of which the set of maximal repeats is
  $\M(S) = \{\mathtt{b}, \mathtt{ab}, \mathtt{bb}, \mathtt{bab}, \mathtt{abab}, \mathtt{ababb}\}$.
  We have
  \begin{itemize}
  \item $\Phi_{S}(\mathtt{bb}) = \{(\underline{\mathtt{b}}, \underline{\mathtt{\$}})\}$ and thus $\phi_{S}(\mathtt{bb}) = 1$, because $\occ_S(\mathtt{bb}) = 3$ and $\occ_S(\mathtt{\underline{b}bb}) = \occ_S(\mathtt{bb\underline{\$}}) = \occ_{S}(\mathtt{\underline{b}bb\underline{\$}}) = 1$;

  \item $\Phi_{S}(\mathtt{bab}) = \{(\mathtt{\underline{b}}, \mathtt{\underline{a}})\}$ and thus $\phi_S(\mathtt{bab}) = 1$, because $\occ_S(\mathtt{bab}) = 4$ and $\occ_S(\mathtt{\underline{b}bab}) = \occ_S(\mathtt{bab\underline{a}}) = \occ_{S}(\mathtt{\underline{b}bab\underline{a}}) = 1$;
    
  \item $\Phi_{S}(\mathtt{abab}) = \{(\mathtt{\underline{b}}, \mathtt{\underline{c}})\}$ and thus $\phi_S(\mathtt{abab}) = 1$, because $\occ_S(\mathtt{abab}) = 3$ and $\occ_S(\mathtt{\underline{b}abab}) = \occ_S(\mathtt{abab\underline{c}}) = \occ_{S}(\mathtt{\underline{b}abab\underline{c}}) = 1$;
  \item $\Phi_S(\mathtt{ababb}) = \{(\mathtt{\underline{c}}, \mathtt{\underline{b}})\}$ and thus $\phi_{S}(\mathtt{ababb}) = 1$, because $\occ_S(\mathtt{ababb}) = 2$ and $\occ_{S}(\mathtt{\underline{c}ababb}) = \occ_S(\mathtt{ababb\underline{b}}) = \occ_{S}(\mathtt{\underline{c}ababb\underline{b}}) = 1$.
  \end{itemize}
  Since the other maximal repeats in $S$ do not have positive net frequencies,
  the list $\mathcal{L}$ stores $\PNF(S) = \{(\mathtt{bb}, 1), (\mathtt{bab}, 1), (\mathtt{abab}, 1), (\mathtt{ababb}, 1)\}$.

  Consider string $\mathtt{b}S = \mathtt{bababbababcababbb\$}$ of which the set of maximal repeats is
  $\M(\mathtt{b}S) = \{\mathtt{b}, \mathtt{ab}, \mathtt{bb}, \mathtt{bab}, \mathtt{abab}, \mathtt{ababb}, \mathtt{babab}\}$.
  For the new maximal repeat $\mathtt{babab} \in \M(\mathtt{b}S) \setminus \M(S)$ we have
  \begin{itemize}
  \item $\Phi_{\mathtt{b}S}(\mathtt{babab}) = \{(\mathtt{\underline{b}, \mathtt{\underline{c}}})\}$ and thus $\phi_{\mathtt{b}S}(\mathtt{babab}) = 1$, because $\occ_{\mathtt{b}S}(\mathtt{babab}) = 2$ and $\occ_{\mathtt{b}S}(\mathtt{\underline{b}babab}) = \occ_{\mathtt{b}S}(\mathtt{babab\underline{c}}) = \occ_{\mathtt{b}S}(\mathtt{\underline{b}babab\underline{c}}) = 1$. This is in Case (ii).
  \end{itemize}
  For the existing maximal repeats in $\M(S) \cap \M(\mathtt{b}S)$, we have
  \begin{itemize}
  \item $\Phi_{\mathtt{b}S}(\mathtt{ababb}) = \{(\mathtt{\underline{b}}, \mathtt{\underline{a}}), (\mathtt{c}, \mathtt{b})\}$ and thus $\phi_{\mathtt{b}S}(\mathtt{abab}) = 2$, because $\occ_{\mathtt{b}S}(\mathtt{ababb}) = 2$, $\occ_{\mathtt{b}S}(\mathtt{\underline{b}ababb}) = \occ_{\mathtt{b}S}(\mathtt{ababb\underline{a}}) = \occ_{\mathtt{b}S}(\mathtt{\underline{b}ababb\underline{a}}) = 1$, and $\occ_{\mathtt{b}S}(\mathtt{cababb}) = \occ_{\mathtt{b}S}(\mathtt{ababbb}) = \occ_{\mathtt{b}S}(\mathtt{cababbb}) = 1$. This is in Case (i);
  \item $\Phi_{\mathtt{b}S}(\mathtt{abab}) = \emptyset$ and thus $\phi_{\mathtt{b}S} (\mathtt{abab})= 0$, because $\occ_{\mathtt{b}S}(\mathtt{abab}) = 3$ and $\occ_{\mathtt{b}S}(\mathtt{abab\underline{c}}) = \occ_{\mathtt{b}S}(\mathtt{\underline{b}abab\underline{c}}) = 1$ but $\occ_{\mathtt{b}S}(\mathtt{\underline{b}abab}) = 2$. This is in Case (iii);
  \item $\Phi_{\mathtt{b}S}(\mathtt{bab}) = \emptyset$ and thus $\phi_{\mathtt{b}S}(\mathtt{bab}) = 0$, because $\occ_{\mathtt{b}S}(\mathtt{bab}) = 4$ and $\occ_{\mathtt{b}S}(\mathtt{\underline{b}bab}) = \occ_{\mathtt{b}S}(\mathtt{\underline{b}bab\underline{a}}) = 1$ but $\occ_{\mathtt{b}S}(\mathtt{bab\underline{a}}) = 2$. This is in Case (iii).
  \end{itemize}
  Since there are no other changes, the elements in the list $\mathcal{L}$ are updated to $\PNF(\mathtt{b}S) = \{(\mathtt{bb}, 1), (\mathtt{ababb}, 2), (\mathtt{babab}, 1)\}$.
\end{example}

\subsection{Suffix Trees and Weiner Links}

For convenience,
we assume that each string $S$ terminates with a unique character $\$$ that does not appear elsewhere in $S$.

A compacted trie is a rooted tree such that
\begin{itemize}
\item each edge is labeled by a non-empty string,
\item each internal node is branching, and
\item the string labels of the out-going edges of each node begin with mutually distinct characters.
\end{itemize}
The \emph{suffix tree}~\cite{Weiner} 
for a string $S$, denoted $\STree(S)$,
is a compacted trie which represents $\Suffix(S)$.
We sometimes identify node $v$ with the substring it represents.


For any string $S$ of length $n$ that terminates with $\$$,
$\STree(S)$ has exactly $n$ leaves,
and has at most $n-1$ internal nodes (including the root)
since every non-leaf node is branching.
By representing each edge label $x$ with an ordered pair $\langle i, j \rangle$
of integers such that $x = S[i..j]$,
$\STree(S)$ can be stored in $O(n)$ space.

We define the \emph{suffix link} of each non-root node $av$
of $\STree(S)$ with $a \in \Sigma$ and $v \in \Sigma^*$,
by $\suflink(av) = v$.
For each node $v$ and $a \in \Sigma$,
we also define the reversed suffix link (a.k.a. \emph{Weiner link})
by $\W{a}{v} = avy$, where $y \in \Sigma^*$ is the shortest string
such that $avy$ is a node of $\STree(S)$.
$\W{a}{v}$ is undefined if $av \notin \Substr(S)$.
A Weiner link $\W{a}{v} = avy$ is said to be \emph{hard}
if $y = \varepsilon$, and \emph{soft} if $y \in \Sigma^+$.


It is clear that every internal node in $\STree(S)$
represents a right-maximal repeat.
Hence, the first characters in the out-edges of each node $u$
are right-extensions of the corresponding right-maximal repeat $u$ in $S$.
Also, for each internal node $u$ that is also a left-maximal repeat (i.e., $u \in \M(S)$),
the characters $a\in \Sigma$, for which $\W{a}{u}$ are defined, are the left-extensions of $u$ in $S$.

\begin{lemma}[\cite{Blumer85}] \label{lem:DAWG_linear}
For any string $S$ of length $n > 2$,
the total number of hard and soft Weiner links in $\STree(S)$ is at most $3n-4$.
\end{lemma}

Throughout the description of our algorithm in Section~\ref{sec:algorithm},
we will assume that the input string $S$ terminates with a unique end-marker $\$$.
We note that this assumption only gives little impact in our NF computation.
By our key lemma (Lemma~\ref{lem:unique_pair}) from the next section,
we have that there can exist at most one character $\alpha$
such that $(\alpha, \$) \in \Phi_{S}(w)$ for a repeat $w$ in $S$.
Moreover, there can exist at most one repeat $w$
that can have $\$$ in the second element
of the pair $(\alpha, \$) \in \Phi_{S}(w)$:
Let $w$ be the longest repeated suffix of $S[1..n-1]$ where $n = |S|$,
and let $\alpha w$ be the shortest non-repeated suffix of $S[1..n-1]$
where $\alpha \in \Sigma$.
Since $S[n] = \$$,
we have $\occ_{S[1..n-1]}(w) = \occ_{S}(w) > 1$ and
$\occ_{S[1..n-1]}(\alpha w) = \occ_{S}(\alpha w) = 1$.
Since $\occ_{S}(w \$) = \occ_{S}(\alpha w \$) = 1$,
$(\alpha, \$) \in \Phi_{S}(w)$ holds.
Note that no other suffixes of $S[1..n-1]$ can satisfy
these properties.

\section{Algorithm} \label{sec:algorithm}

Let $S$ be the input string.
As in the previous work~\cite{GuoUWZ24},
a node $w$ in $\STree(S)$ stores the NF value $\phi_S(w)$ iff $\phi_S(w) \geq 1$.

\subsection{Key Lemma}

Our algorithm is based on the following lemma:

\begin{lemma}[Unique Character Pairs] \label{lem:unique_pair}
  Let $w$ be a repeat in a string $S$.
  Let $(\alpha,\beta) \in \Phi_S(w)$.
  Then, for any character $\beta' \neq \beta$,
  $(\alpha, \beta') \notin \Phi_S(w)$.
  Similarly, for any character $\alpha' \neq \alpha$,
  $(\alpha', \beta) \notin \Phi_S(w)$.
\end{lemma}

\begin{proof}
  For a contradiction assume that $(\alpha, \beta') \in \Phi_S(w)$.
  Then, $\occ_S(\alpha w \beta) = 1$ and $\occ_S(\alpha w \beta') = 1$.
  Since $\beta \neq \beta'$, $\occ_S(\alpha w) \geq 2$.
  However, since $(\alpha,\beta) \in \Phi_S(w)$, we have $\occ_S(\alpha w) = 1$,
  a contradiction.
  The case for $(\alpha', \beta)$ is symmetric.
\end{proof}

Lemma~\ref{lem:unique_pair} implies that for any character $\alpha \in \Sigma$
that is a left-extension of a (maximal) repeat $w$ in $S$,
there is at most one character $\beta \in \Sigma$ that contributes to
the positive NF of $w$ in $S$.
Thus the following corollary is immediate:
\begin{corollary} \label{coro:NF_upperbound}
  For any repeat $w$ in a string $S$,
  $\phi_S(w) \leq \min\{|\Sigma^L_w|, |\Sigma^R_w|\}$,
  where $\Sigma^L_w$ and $\Sigma^R_w$ are the sets of left-extensions
  and right-extensions of $u$ in $S$, respectively.
\end{corollary}

Since the total number of right-extensions of the maximal repeats
is most $2n-2$ for any string of length $n \geq 2$ (c.f.~\cite{Blumer87}),
the next lemma immediately follows from Lemma~\ref{lem:NF_MR} and Corollary~\ref{coro:NF_upperbound}.
\begin{lemma} \label{lem:total_NF_upperbound}
For any string $S$ of length $n$,   
$\sum_{w \in \Substr(S)} \phi_S(u) \leq 2n-2$.
\end{lemma}
This bound is weaker than
the bound $\sum_{w \in \Substr(S)} \phi_S(w) \leq n$ given by
Guo et al.~\cite{GuoEWZ24},
but it is sufficient for our purpose.
Namely, if needed, one can explicitly store the character pairs in $\Phi_S(w)$
in $O(n)$ total space instead of the counts $\phi_S(w)$.
Our algorithm to follow can easily be modified to store
the character pairs in $\Phi_S(w)$ as well, in the same complexity.
  
\subsection{Recalling Weiner's Suffix Tree Construction}

In this subsection, we briefly recall
how Weiner's algorithm updates $\STree(S)$ for a string $S$
to $\STree(aS)$ with a new character $a$ that is prepended to $S$.

Suppose that all the Weiner links have also been computed for
the nodes of $\STree(S)$.
Since our input string $S$ terminates with $\$$ at its right end,
all the suffixes of $S$ are represented by the leaves of $\STree(S)$.

To update $\STree(S)$ into $\STree(aS)$,
it is required to insert a new leaf that represents the new suffix $aS$.
To insert a new leaf for $aS$,
Weiner's algorithm finds the locus of the longest repeating prefix $av$ of $aS$ on $\STree(S)$.
Then, $av$ will be the parent of the new leaf for $aS$.

\begin{figure}[tbh]
  \centering
  \raisebox{3mm}{
    \includegraphics[scale=0.38]{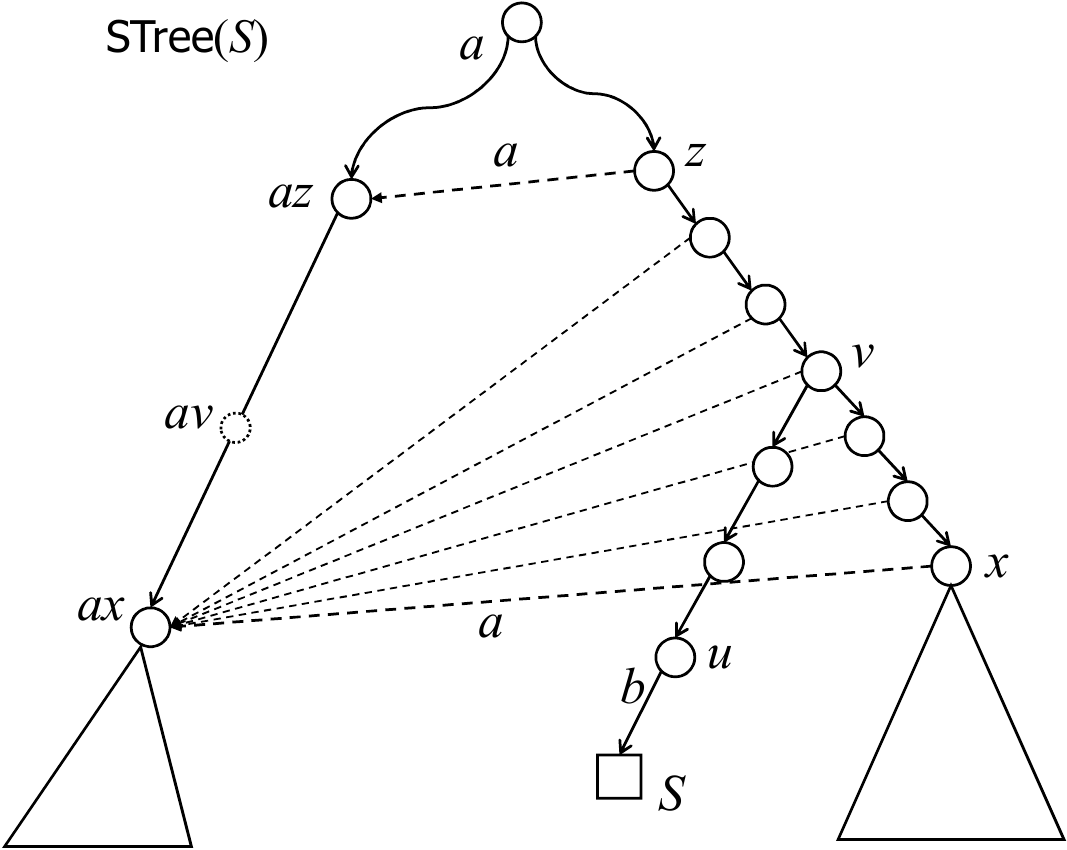}
  }  
  \hfill
  \includegraphics[scale=0.38]{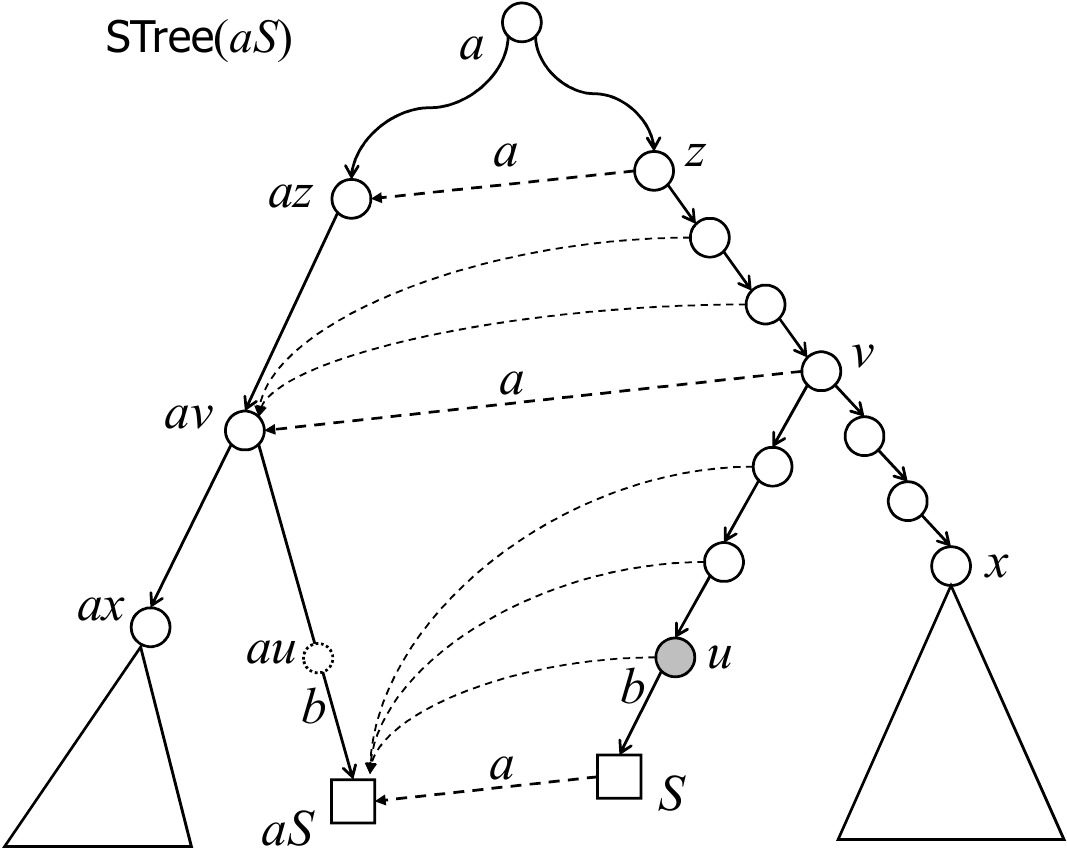}
  \caption{A snapshot of Weiner's suffix tree construction algorithm. Solid arcs represent edges, and broken thick (resp. thin) arcs represent hard (resp. soft) Weiner links. Left: By climbing up the path from leaf $S$, we find the lowest ancestor $v$ that has a soft Weiner link $\W{a}{v} = ax$ that is labeled by $a$, and the lowest ancestor $z$ that has a hard Weiner link $\W{a}{z} = az$ that is labeled by $a$. Right: Split the edge from node $az$ to $ax$ at the locus for $av$, and insert a new leaf for $aS$ from $av$. If the parent $u$ of $S$ is not $v$, then the NF of the parent $u$ of $S$ is increased by 1, since $(a,b) \in \Phi_{aS}(u) \setminus \Phi_{S}(u)$.}
  \label{fig:Weiner}
\end{figure}

In Weiner's algorithm, we climb up the path from the leaf for $S$
and find the lowest ancestor $v$ such that the Weiner link $\W{a}{v}$ with
the new character $a$ is defined.
See also Figure~\ref{fig:Weiner} for illustration.
There are two cases:

(1) If $\W{a}{v}$ is a soft Weiner link, then let $ax$ be the destination node
of this Weiner link.
Note that $av$ is a proper prefix of $ax$.
We continue climbing up the path until finding the lowest ancestor $z$
that has a hard Weiner link $\W{a}{z}$ with $a$.
We move to node $az$ using the hard Weiner link $\W{a}{z} = az$.
The locus for $av$ is on the edge from $az$ to $ax$.
We split this edge at the locus for $av$, and insert the new leaf for $aS$.
Since $av$ is a node in the updated suffix tree $\STree(aS)$,
the Weiner link $\W{a}{v}$ becomes hard and points to $av$.
The soft Weiner links of the nodes between $v$ and $z$ exclusive
are redirected to $av$.
Finally, we copy all (soft or hard) Weiner links of node $ax$
to the new node $av$,
namely, if $\W{d}{ax} = y$ for character $d \in \Sigma$
and node $y = dax$, then we set $\W{d}{av} = y$.

(2) If $\W{a}{v}$ is a hard Weiner link,
then we move to node $av$ using the hard Weiner link $\W{a}{v} = av$
(this is the case where $v = z$ in Figure~\ref{fig:Weiner}).
We create a new leaf for $aS$ as a new child of $av$.

In both cases, 
we add new Weiner links from the nodes in the path between $S$ and $v$
($S$ inclusive and $v$ exclusive)
to the new leaf $aS$ with the character $a$.

\begin{lemma}[\cite{Weiner,Blumer85}] \label{lem:Weiner_time}
  The number of nodes visited during the update from $\STree(S)$ to $\STree(aS)$
  is amortized $O(1)$,
  and the total time cost for updating $\STree(S)$ into $\STree(aS)$ is amortized $O(\log \sigma)$.
\end{lemma}

The $\log$ factor in the update time is the cost for accessing each Weiner link that is labeled by the new character $a$.

\subsection{Computing $\PNF(aS) \triangle \PNF(S)$}

Suppose that we have built $\STree(S)$ for string $S$ of length $n$ augmented with soft and hard Weiner links,
and that we have computed the list $\PNF(S)$.
For each node $w$ of $\STree(S)$ with $\phi_{S}(w) \geq 1$,
$(w, \phi_{S}(w)) \in \PNF(S)$ is stored in $w$.
By Lemma~\ref{lem:NF_MR} (or alternatively by Lemma~\ref{lem:total_NF_upperbound}),
we can store $\PNF(S)$ for all nodes $w$ of $\STree(S)$
in total $O(n)$ space.

We update $\STree(S)$ to $\STree(aS)$ with new character
$a$ using Weiner's algorithm, as described above.

By recalling the definitions of $\Phi_S(w)$
and $\phi_S(w)$ respectively in Equation~\ref{eqn:Phi} and in Equation~\ref{eqn:phi} for a string $w$,
we can reduce Problem~\ref{prob:online_NF_right_to_left}
to computing ordered character pairs $(\alpha,\beta) \in \Sigma \times \Sigma$ as described in the following lemma:
\begin{lemma} \label{lem:NF_iff}
  For any characters $\alpha,\beta \in \Sigma$
  and strings $w,S \in \Sigma^*$,
  $(\alpha,\beta) \in \Phi_S(w)$ iff
  \begin{itemize}
  \item[(1)] $w$ is an internal branching node of $\STree(S)$,
  \item[(2)] $w$ has a soft Weiner link $\W{\alpha}{w}$ that points to a leaf $\alpha wt$, where $t \in \Sigma^+$ and $\beta = t[1]$, and
  \item[(3)] $w$ has an out-edge that begins with $\beta$ and leads to a leaf.
  \end{itemize}
\end{lemma}
\begin{proof}
  Suppose $(\alpha,\beta) \in \Phi_S(w)$ holds.
  Then, by definition, $w$ is a repeat in $S$.
  Further, by Lemma~\ref{lem:NF_MR} and by our assumption that 
  $S$ terminates with a unique end-marker $\$$,
  (1) $w$ is an internal branching node in $\STree(S)$.
  By assumption we have $\occ_{S}(w \beta) = 1$,
  and thus (3) $w$ has an out-edge that begins with $\beta$ and leads to a leaf.
  Also, by assumption we have $\occ_{S}(\alpha w) = \occ_S(\alpha w \beta) = 1$.
  This implies that $\alpha w$ is on an edge leading to a leaf $\alpha wt$
  with $t \in \Sigma^+$ and $t[1] = \beta$.
  Also, $\alpha w \beta$ is on the same edge if $|t| > 1$,
  or it is a leaf if $|t| = 1$.
  Since $w$ is preceded by $\alpha$ in $S$,
  $w$ has a Weiner link $\W{\alpha}{w}$ labeled $\alpha$.
  Since $\alpha w$ is on an edge,
  $\W{\alpha}{w} = \alpha w t$ with $|t| \geq 1$,
  and thus it is a soft Weiner link.
  This leads to (2).

  Suppose (1), (2), and (3) all hold.
  Then, it is clear that $\occ_{S}(\alpha w) = \occ_{S}(w \beta) = \occ_{S}(\alpha w \beta) = 1$ and $\occ_{S}(w) \geq 2$.
\end{proof}

\begin{figure}[tbh]
  \centering
  \includegraphics[scale=0.5]{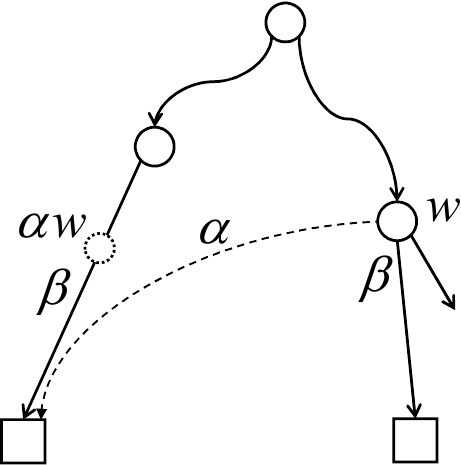}
  \caption{Illustration for Lemma~\ref{lem:NF_iff},
    where the rectangles represent leaves,
    the solid circles represent internal nodes,
    and the dashed circle represents an implicit node (a locus on an edge).
    The dashed arc represents a soft Weiner link.}
  \label{fig:NF_iff}
\end{figure}

See Figure~\ref{fig:NF_iff}
that illustrates the situation in Lemma~\ref{lem:NF_iff}.
The intuition of our online NF computation
based on Weiner's algorithm is as follows:
Each time the suffix tree is updated,
we examine whether a substructure
described in Figure~\ref{fig:NF_iff} newly appears in the suffix tree,
and whether an existing substructure described in Figure~\ref{fig:NF_iff} disappears from the suffix tree.
More formally, by the definitions of $\PNF(S)$ and $\PNF(aS)$
and by Lemma~\ref{lem:NF_iff},
it suffices for us to check the following types of nodes,
where Type (i), Type (ii), and Type (iii)-1\&2
correspond to Case (i), Case (ii), and Case (iii) of Section~\ref{sec:online_NF_def}, respectively:
\begin{description} 
  \item[Type (i)] Existing maximal repeat $w \in \M(aS) \cap \M(S)$ in $aS$ that obtains
    a new left-extension $\alpha$
    and an existing right-extension $\beta$
    such that $\occ_{aS}(\alpha w \beta) = \occ_{aS}(\alpha w) = \occ_{aS}(w \beta) = 1$.
    We will later show that the parent $u$ of the leaf for $S$
    is the only candidate for this type.

  \item[Type (ii)] New maximal repeat $w \in \M(aS) \setminus \M(S)$ in $aS$
  that has a left-extension $\alpha$ and a right-extension $\beta$ with
  $\occ_{aS}(\alpha w \beta) = \occ_{aS}(\alpha w \beta) = \occ_{aS}(w \beta) = 1$.
  We will later show that the parent $av$ of the new leaf for $aS$ is the only candidate for this type.

  \item[Type (iii)-1] Existing maximal repeat $w \in \M(aS) \cap \M(S)$ in $aS$
    that has a left-extension $\alpha$ and a right-extension $\beta$
    with $\occ_{S}(\alpha w) = \occ_{S}(w \beta) = \occ_{S}(\alpha w \beta) = 1$ and $\occ_{aS}(\alpha w) > 1$.
    Notice that $w \in \M(S)$ holds for each node $w$ of Type (iii)-1,
    since $\phi_{S}(w) \geq 1$ for Case (iii) in Section~\ref{sec:online_NF_def}.
    We will later show that the lowest ancestor $v$ of the leaf for $S$
    that has a soft Weiner link labeled $a$ in $\STree(S)$
    is the only candidate for this type.
    
  \item[Type (iii)-2] Existing maximal repeat $w \in \M(aS) \cap \M(S)$ in $aS$
    that has a left-extension $\alpha$ and 
    a right-extension $\beta$ with
    $\occ_{S}(\alpha w) = \occ_{S}(w \beta) = \occ_{S}(\alpha w \beta) = 1$ and $\occ_{aS}(w \beta) > 1$.
    Notice that $w \in \M(S)$ holds for each node $w$ of Type (iii)-2,
    since $\phi_{S}(w) \geq 1$ for Case (iii) in Section~\ref{sec:online_NF_def}.
    We will later show that
    the node $az$ that is the destination node of the hard Weiner link
    from $z$ labeled $a$ is the only candidate for this type.
    
\end{description}

In Sections~\ref{sec:type-i}-\ref{sec:type-iii-2},
we show how to deal with the nodes of Types (i), (ii), (iii)-1, and (iii)-2.
Our algorithm simply uses Weiner's algorithm
that updates $\STree(S)$ to $\STree(aS)$,
and it suffices for us to perform simple verifications
for the visited nodes based on Lemma~\ref{lem:NF_iff}.

\subsection{Dealing with Type (i) node} \label{sec:type-i}
Here we show that the number of nodes of Type (i) is at most one per iteration.

\begin{lemma} \label{lem:type_i}
Let $v$ be the lowest ancestor of $S$ such that the Weiner link $\W{a}{v}$ with
the new character $a$ is defined,
and $u$ be the parent of the leaf representing $S$ in $\STree(S)$.
If $u \neq v$,  
then $u$ is a node of Type (i)
and we have $(a,b) \in \Phi_{aS}(u) \setminus \Phi_{S}(u)$,
where $b$ is the first character of the edge label
from $u$ to $S$ (see also Figure~\ref{fig:Weiner}).
In addition, $u$ is the only possible node of Type (i).
\end{lemma}

\begin{proof}
  It follows from the definition of Type (i) nodes that
  for any Type (i) nodes,
  the first character $\alpha$ in the pair $(\alpha, \beta)$
  has to be the new character $a$
  that is prepended to $S$.
  If $(a, \beta) \in \Phi_{aS}(u)$ for the parent $u$ of the leaf $S$,
  then this character $\beta$ must be unique by Lemma~\ref{lem:unique_pair}.
  Further by Lemma~\ref{lem:NF_iff},
  this $\beta$ is the first character $b$ of the edge label from $u$ to $S$
  (see also the right diagram of Figure~\ref{fig:Weiner}).
  Now consider the case that $u \neq v$,
  where $v$ is the lowest ancestor of $S$
  that has a (soft or hard) Weiner link labeled $a$ in
  $\STree(S)$ before the update.
  Since $u \neq v$,
  $u$ obtains a new soft Weiner link $\W{a}{u}$ labeled $a$
  toward the leaf $aS$
  in $\STree(aS)$ after the update 
  (see the right diagram of Figure~\ref{fig:Weiner}).
  Then the conditions of Lemma~\ref{lem:NF_iff}
  are satisfied for the node $u$ and the character pair $(a,b)$ in $\STree(aS)$.
  Therefore, $(a,b) \in \Phi_{aS}(u)$ holds in this case.
  This can also be observed by the corresponding new substructure in form of Figure~\ref{fig:NF_iff} that appears in $\STree(aS)$.
  Clearly $(a,b) \notin \Phi_S(u)$ because $\occ_{S}(aub) = 0$.
  Hence $u$ is of Type (i) when the aforementioned conditions are satisfied.

  Next, we show that $u$ is the only possible node of Type (i).
  Let $u'$ be any node on the path between $v$ and $u$, exclusive.
  Notice that these nodes $u'$ are the only nodes
  other than $u$
  that obtain a new left-extension $a$ (see also Figure~\ref{fig:Weiner}).
  Assume for a contrary that
  $u'$ is of Type (i),
  which implies that 
  $\occ_{aS}(au') = 1$ and $\occ_{S}(au') = 0$.
  Let $b'$ be the only character that immediately follows $au'$ in $aS$,
  namely $\occ_{aS}(au'b') = 1$.
  Note that such character $b'$ must exist for each $u'$.
  It follows from Lemma~\ref{lem:NF_iff} that
  it suffices to check whether $u'$ has an out-edge
  labeled $b'$ that leads to a leaf.
  However, this is not the case:
  $au'b'$ is a prefix of $aS$,
  and thus $u'b'$ is a prefix of $S$.
  Thus, $b'$ is the first character in the path label from $u'$ to the leaf $S$.
  Since $u'$ is a non-parent ancestor of the leaf $S$,
  the out-edge from $u'$ that begins with $b'$ does not lead to a leaf.
  Thus, $u'$ cannot be a Type (i) node by Lemma~\ref{lem:NF_iff}.
\end{proof}

This node $u$ is found in the process of climbing up
the path from the leaf for $S$ toward its lowest ancestor $z$
that has a hard Weiner link $\W{a}{z}$ for $a \in \Sigma$.
The character $b$ can easily be obtained
by using the string depths of $u$ and $S$.
When all the conditions in Lemma~\ref{lem:type_i} are satisfied,
then $u$ is of Type (i).
We then update $\Phi_{aS}(u) \leftarrow \Phi_{S}(u) \cup \{(a,b)\}$.
We add $(u, 1)$ to the list $\mathcal{L}$ if $\phi_{S}(u) = 0$, 
and replace $(u, \phi_{S}(u))$ with $(u, \phi_{S}(u)+1) = (u, \phi_{aS}(u))$ in $\mathcal{L}$ otherwise.

\subsection{Dealing with Type (ii) node}

Weiner's algorithm inserts at most one new internal node per iteration,
which is the parent $av$ of the new leaf representing $aS$ (see also Figure~\ref{fig:Weiner_leaf} for an illustration).
Thus, $av$ is the only candidate for a node $w$ of Type (ii),
implying that the number of nodes of Type (ii) is at most one per iteration.
The following lemma shows the only case where the node $av$ obtains a new pair of characters in $\Phi_{aS}(av)$.

\begin{figure}[tbh]
  \centering
  \raisebox{3mm}{
    \includegraphics[scale=0.34]{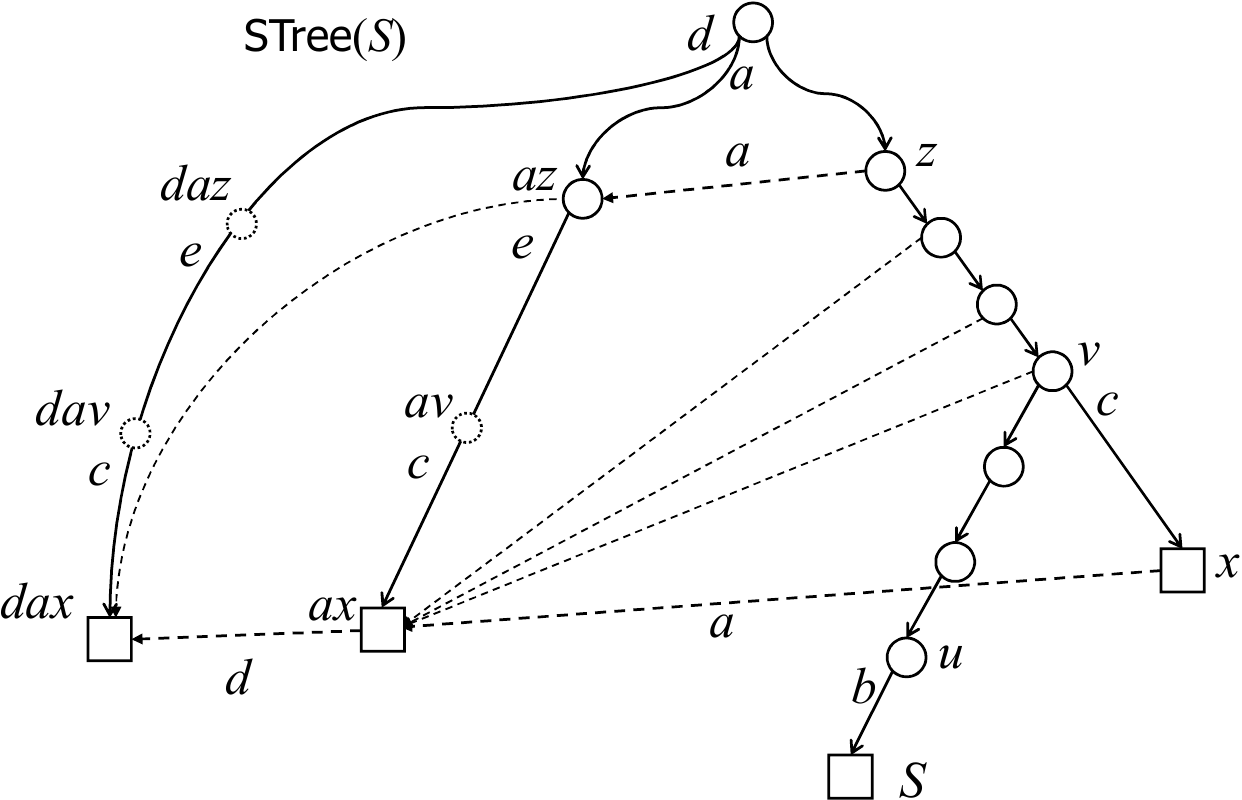}
  }  
  \hfill
  \includegraphics[scale=0.34]{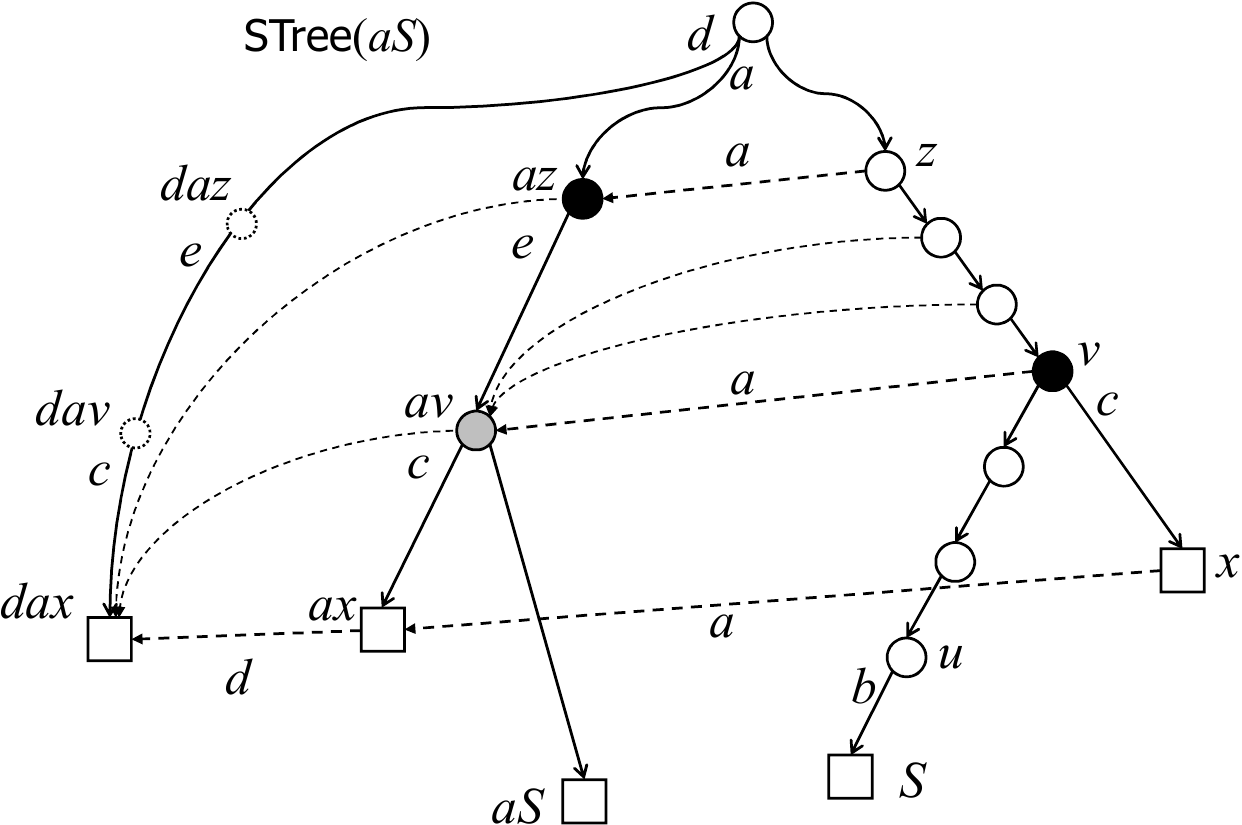}
  \caption{A snapshot of Weiner's suffix tree construction algorithm
    in the case where $ax$ is a leaf.
    If the leaf $x$ is a child of $v$,
    then $(a,c) \in \Phi_{S}(v) \setminus \Phi_{aS}(v)$.
    Let $d$ be a unique character that immediately precedes $ax$ in $S$.
    Then $(d,c) \in \Phi_{aS}(av) \setminus \Phi_{S}(av)$.
    If $daz$ in an implicit node on the edge leading to the leaf $dax$,
    then $(d, e) \in \Phi_{S}(az) \setminus \Phi_{aS}(az)$.}
  \label{fig:Weiner_leaf}
\end{figure}

\begin{lemma} \label{lem:type-ii}
  Let $v$ be the lowest ancestor of leaf $S$ in $\STree(S)$ such that $\W{a}{v}$ is defined (see also Figure~\ref{fig:Weiner_leaf}),
  and let $ax$ be the destination node of the Weiner link $\W{a}{v}$ from $v$ labeled $a \in \Sigma$ in $\STree(S)$ before the update.
  Let $c$ be the first character of the edge from $av$ to $ax$.
  If $av$ is on an edge in $\STree(S)$ and the node $ax$ is a leaf,
  then $(d, c) \in \Phi_{aS}(av) \setminus \Phi_{S}(av)$,
  where $d \in \Sigma$ is the unique character which immediately precedes $ax$ in $S$. 
\end{lemma}

\begin{proof}
  Recall that $av$ becomes an internal node
  as the parent of a new leaf for $aS$ in $\STree(aS)$ after the update.
  Since Weiner's algorithm inserts at most one internal node
  and it is $av$ in this case,
  the string $dav$ remains on an edge in $\STree(aS)$ after the update.
  Thus, $av$ obtains a soft Weiner link labeled $d \in \Sigma$ toward node $dax$.
  Since $ax$ is a leaf, $dax$ is also a leaf.
  For the string $t$ such that $dax = davt$,
  we have $t[1] = c$.
  Thus, we have $(d,c) \in \Phi_{aS}(av)$ by Lemma~\ref{lem:NF_iff}.
  This can also be observed by the corresponding new substructure in form of Figure~\ref{fig:NF_iff} that appears in $\STree(aS)$.
  Due to Lemma~\ref{lem:unique_pair},
  $d$ is the unique character pair for $c$ that can possibly be a member of $\Phi_{aS}(av)$.
  Since $c$ was the only character immediately following $av$ in $S$,
  $(d,c)$ is the only character pair that could be a member for $\Phi_{aS}(av)$.
  Since $\occ_{S}(av) = 1$, $av$ was not a repeat in $S$ before the update.
  Thus $\phi_{S}(av) = 0$ and hence $(d,c) \notin \Phi_{S}(av)$ holds for the string $S$ before the update.

  If $ax$ is not a leaf, then $\occ_{aS}(avc) \geq \occ_{aS}(ax) > 1$,
  implying $(d, c) \notin \Phi_{aS}(av)$.
  Also, if $av$ is a node in $\STree(S)$,
  then $av$ was already a maximal repeat in $S$
  and thus $av$ is not a Type (ii) node by definition.
\end{proof}

  The locus for $av$ is found via the Weiner link $\W{a}{v}$ from the node $v$
  during the update from $\STree(S)$ to $\STree(aS)$.
  We remark that $av$ is on an edge in $\STree(S)$
  if and only if $\W{a}{v}$ is a soft Weiner link in $\STree(S)$ before the update.
  Further, we can easily check if $ax$ is a leaf or not.
  The character $c$ can be obtained by using
  the string depths of the nodes $av$ and $ax$
  (see also Figure~\ref{fig:Weiner_leaf}).
  The character $d$ can be found easily
  (e.g. by using the unique Weiner link from leaf $ax$.)

  If the conditions in Lemma~\ref{lem:type-ii} are satisfied
  for the nodes $av, ax$ and characters $d,c$,
  the conditions in Lemma~\ref{lem:NF_iff} are also satisfied.
  Then, we update $\Phi_{aS}(av) \leftarrow \{(d,c)\}$,
  and add $(av, \phi_{aS}(av)) = (av, 1)$ to the list $\mathcal{L}$.

\subsection{Dealing with Type (iii)-1 node}

Here we show that the number of nodes of Type (iii)-1 is also at most one per iteration.

\begin{lemma}
Let $v$ be the lowest ancestor of leaf $S$ in $\STree(S)$ such that $\W{a}{v}$ is defined (see also Figure~\ref{fig:Weiner_leaf}),
  and let $ax$ be the destination node of $\W{a}{v}$ from $v$ labeled $a \in \Sigma$ in $\STree(S)$ before the update.
Let $c$ be the first character of the edge from $v$ to $x$.
Then $v$ is the only possible node of Type (iii)-1.
Also, $(a,c)$ is the only possible character pair
such that $(a,c) \in \Phi_{S}(v) \setminus \Phi_{aS}(v)$.
\end{lemma}

\begin{proof}
  Consider the case that $\occ_S(av) = 1$ and $\occ_{aS}(av) > 1$.
  Since $\occ_S(av) = 1$, we have $\occ_S(ax) = 1$ and thus $ax$ is a leaf.
  Since $ax \in \Suffix(S)$, $x \in \Suffix(S)$ and hence $x$ terminates with $\$$.
  Since $v$ is an internal node, $v$ cannot end with $\$$ and thus
  $v$ is a proper prefix of $x$.
  Thus $av$ is on the edge leading to the leaf $ax$ in $\STree(S)$ before the update.
  Since $x$ terminates with $\$$,
  $x$ is also a leaf.
  Let $c$ be the first character of the path label from $v$ toward $x$.
  If $x$ is not a child of $v$,
  then $\occ_{S}(vc) > 1$ and thus $v$ cannot be of Type (iii)-1.
  Consider the case where $x$ is a child of $v$.
  In this case, $\occ_{S}(av) = \occ_{S}(vc) = \occ_{S}(avc) = 1$
  (see also the left diagram of Figure~\ref{fig:Weiner_leaf}),
  and thus we have $(a,c) \in \Phi_S(v)$.
  After the update,
  $av$ becomes a branching node in $\STree(aS)$,
  meaning that $\occ_{aS}(av) > 1$
  (see also the right diagram of Figure~\ref{fig:Weiner_leaf}).
  By Lemma~\ref{lem:NF_iff}, $(a,c) \notin \Phi_{aS}(v)$ holds for node $v$.
  This can also be observed with the corresponding existing substructure in form of Figure~\ref{fig:NF_iff} that disappears in $\STree(aS)$.
  Overall, if the aforementioned conditions are satisfied,
  then $v$ is a node of Type (iii)-1.
  Since $\occ_{S}(avc) = 1$,
  $c$ is the only right-extension of $av$ in $S$.
  It now follows from Lemma~\ref{lem:unique_pair} that
  $a$ is the only pairing character for $c$,
  and hence $(a,c)$ is the only possible character pair
  such that $(a,c) \in \Phi_S(v) \setminus \Phi_{aS}(v)$.
    

  Next, we show that $v$ is the only possible node that is of Type (iii)-1.
  Let $v'$ be any node in the path between $v$ and $z$ exclusive,
  where $z$ is the lowest ancestor of the leaf $S$ that has a hard Weiner link
  $\W{a}{z}$ labeled $a$ in $\STree(S)$ before the update.
  Note that those nodes $v'$ are the only nodes other than $v$
  in the path from the root to $S$,
  which satisfy $\occ_{S}(av') = 1$ and $\occ_{S}(v') > 1$.
  However, since the descendant $v$ of $v'$ is an internal node,
  we have $\occ_{S}(v'c') > 1$ and thus
  $(a, c') \notin \Phi_{aS}(v')$ holds,
  where $c'$ is the only possible paring character for character $a$ with node $v'$ due to Lemma~\ref{lem:unique_pair}.
  Thus $v$ is the only possible node of Type (iii)-1.
\end{proof}

The node $v$ is found by climbing up the path
from the leaf for $S$,
during the update from $\STree(S)$ to $\STree(aS)$.
The node $ax$ is obtained by the Weiner link
$\W{a}{v}$ from $v$,
and the node $x$ is obtained by
the reversed Weiner link (i.e. the suffix link) from $ax$.
The character $c$ is obtained by using
the string depths of the nodes $av$ and $ax$.

Recall that $z$ is the lowest ancestor of the leaf $S$
that has a hard Weiner link $\W{a}{z}$ labeled $a$ in $\STree(S)$
before the update.
If $v \neq z$, then $av$ is on an edge.
We can easily check whether $x$ and $ax$ are leaves
and $x$ is a child of $v$ in $\STree(S)$.
Notice that the first condition $v \neq z$ is critical,
since if $v = z$, then $av$~($= az$) is an internal node in $\STree(S)$
and thus $\occ_{S}(av) > 1$.
Now we have two sub-cases:
\begin{itemize}
\item  If $\phi_{S}(v) = 1$ (namely $\Phi_{S}(v) = \{(a,c)\}$),
then we get $\phi_{aS}(v) = 0$ (namely $\Phi_{aS}(v) = \emptyset$)
after the update,
and thus $(u, \phi_{S}(v))$ is removed from the list $\mathcal{L}$.
\item
If $\phi_{S}(v) > 1$, then
we update $\Phi_{aS}(v) \leftarrow \Phi_{S}(v) \setminus\{(a,c)\}$
and replace $(v, \phi_{S}(v))$ with $(v, \phi_{S}(v)-1) = (v, \phi_{aS}(v))$
in $\mathcal{L}$.
\end{itemize}

\subsection{Dealing with Type (iii)-2 node} \label{sec:type-iii-2}

Here we show that the number of nodes of Type (iii)-2 is also at most one per iteration.

\begin{lemma}
  Let $z$ be the lowest ancestor of leaf $S$ in $\STree(S)$ such that $\W{a}{z}$ is hard  (see also Figure~\ref{fig:Weiner_leaf}).
  Let $v$ be the lowest ancestor of leaf $S$ in $\STree(S)$ such that $\W{a}{v}$ is defined,
  and let $ax$ be the destination node of $\W{a}{v}$ from $v$ labeled $a \in \Sigma$ in $\STree(S)$ before the update.
  Let $e$ be the first character in the path from $az$ to $ax$.
Then $az$ is the only possible node of Type (iii)-2.
Also, $(d,e)$ is the only possible character pair
such that $(d,e) \in \Phi_{S}(az) \setminus \Phi_{aS}(az)$,
  where $d \in \Sigma$ is the unique character which immediately precedes $ax$ in $S$.
\end{lemma}

\begin{proof}
  Recall that $az$ is the parent of $ax$ in $\STree(S)$ before the update.  
  Consider the case that $\occ_S(aze) = 1$, where $ax$ is a leaf in $\STree(S)$.
  Also, consider the case that $\occ_{S}(daz) = 1$,
  which implies $\occ_{S}(daze) = 1$.
  Therefore, $(d,e) \in \Phi_{S}(az)$ by Lemma~\ref{lem:NF_iff}.
  Consider the case where $av$ becomes a branching node in $\STree(aS)$ after the update.
  Then, $av$ becomes a child of $az$, meaning that $\occ_{aS}(aze) > 1$.
  Thus $(d, e) \notin \Phi_{aS}(az)$ holds for $az$.
  Overall, if the aforementioned conditions are satisfied,
  then $az$ is a node of Type (iii)-2.
  This can also be observed with the corresponding existing substructure in form of Figure~\ref{fig:NF_iff} that disappears in $\STree(aS)$.

  Since Weiner's algorithm inserts at most one new internal node,
  $az \in \M(S)$ and $e \in \Sigma$ are the only maximal repeat and character
  that can satisfy $\occ_{S}(aze) = 1$ and $\occ_{aS}(aze) > 1$.
  By Lemma~\ref{lem:unique_pair},
  $(d,e)$ is the the only possible character pair
such that $(d,e) \in \Phi_{S}(az) \setminus \Phi_{aS}(az)$.
\end{proof}

The node $az$ is found via the hard Weiner link $\W{a}{z}$ in $\STree(S)$
from $z$,
during the update from $\STree(S)$ to $\STree(aS)$.
The character $e$ is found
by using the string depths of the nodes $az$ and $ax$.
We can find the character $d$
and can check whether $ax$ is a leaf or not
analogously to the case of Type (iii)-1 nodes.
The rest is also analogous to Type (iii)-1 nodes.

\subsection{Complexities}

\begin{theorem}
  For a string $S$ of length $n$ and a character $a \in \Sigma$,
  Weiner's algorithm solves Problem~\ref{prob:online_NF_right_to_left}
  of online NF computation in $O(\log \sigma)$ amortized time per character
  using $O(n)$ space.
\end{theorem}

\begin{proof}
  It is known that Weiner's algorithm finds
  the nodes $u$, $v$, and $z$ in $O(1)$ amortized time for
  a new character $a$ that is prepended to $S$ (Lemma~\ref{lem:Weiner_time}).
  It follows from our previous discussions that
  given the nodes $u$, $v$, and $z$,
  we can compute $\PNF(aS) \triangle \PNF(S)$
  in $O(\log \sigma)$ amortized time,
  where the $\log \sigma$ factor is the cost to access the Weiner links
  from nodes $v$ and $z$.

  By Lemma~\ref{lem:DAWG_linear} the number of soft and hard Weiner links is $O(n)$.
  Since the suffix tree takes $O(n)$ space,
  the total space requirement is $O(n)$.
\end{proof}

Since all positive values of NFs are stored and maintained
in the respective nodes of the online suffix tree,
we obtain the following:

\begin{corollary}
  There exists an online data structure that can be maintained in
  $O(n \log \sigma)$ total time with $O(n)$ space, and can answer
  \textsc{Single-NF} queries in $O(m \log \sigma)$ time,
  where $m$ is the query pattern length.
\end{corollary}

By maintaining $\mathcal{L}$ as a doubly-linked list
such that each element is associated to the corresponding node in the suffix tree, we obtain the following:

\begin{corollary}
  There exists an online data structure that can be maintained in
  $O(n \log \sigma)$ total time with $O(n)$ space,
  and can report all answers to the \textsc{All-NF} problem in output optimal
  $O(|\PNF(S)|)$ time.
\end{corollary}

Guo et al.~\cite{GuoEWZ24} also considered a version of the \textsc{All-NF} problem, named \textsc{All-NF-Extract},
that requires us to compute the product $|u| \cdot \phi_S(u)$ for all pairs $(u, \phi_S(u)) \in \PNF(S)$.
Guo et al.~\cite{GuoEWZ24} proposed an offline algorithm that works in $O(n \log \delta)$ time for \textsc{All-NF-Extract} for integer alphabets,
where $\delta$ denotes the \emph{substring complexity}~\cite{KociumakaNP23} of the input string.
Our online algorithm can readily be extended to the \textsc{All-NF-Extract} problem:

\begin{corollary}
  There exists an online data structure that can be maintained in
  $O(n \log \sigma)$ total time with $O(n)$ space,
  and can report all answers to the \textsc{All-NF-Extract} problem in
  output optimal $O(|\PNF(S)|)$ time.
\end{corollary}

\section*{Acknowledgements}
This work is supported by JSPS KAKENHI Grant Numbers
JP23K24808 and JP23K18466.
The author thanks Takuya Mieno for fruitful discussions.

\bibliographystyle{abbrv}
\bibliography{ref}

\end{document}